\documentclass[10pt,letterpaper,reqno]{amsart}

\usepackage{amsmath,amsthm,amssymb,latexsym} 
\usepackage{fixmath} 
\usepackage{fullpage} 
\usepackage{enumerate} 
\usepackage{graphicx} 
\usepackage{color} 
\usepackage{accents} 
\usepackage[all]{xy} 

\usepackage[pdftex,bookmarks,colorlinks,breaklinks]{hyperref}  
\definecolor{dullmagenta}{rgb}{0.4,0,0.4}   
\definecolor{darkblue}{rgb}{0,0,0.4}
\definecolor{mygrey}{gray}{0.5} 
\hypersetup{linkcolor=red,citecolor=blue,filecolor=dullmagenta,urlcolor=darkblue} 

\newcommand{\adzham}{\href{mailto:adzham@unco.edu}{\texttt{adzham@unco.edu}}}
\newcommand{\comment}[1]{}
\newtheorem{theorem}{Theorem}
\newtheorem{corollary}[theorem]{Corollary}
\newtheorem{lemma}[theorem]{Lemma} 

\theoremstyle{definition}
\newtheorem{definition}{Definition}

\numberwithin{theorem}{section}
\numberwithin{equation}{section}
\numberwithin{definition}{section}

\theoremstyle{remark}

\newenvironment{notation}{\noindent$\triangleleft$ {\bf
Notation:}}{$\triangleright$\medskip}

\begin{document}

\title[Discrete Lagrangian Dynamics]{On the Lagrangian Structure of 
the Discrete Isospectral and Isomonodromic Transformations}
\author[A. Dzhamay]{Anton Dzhamay} 
\address{School of Mathematical Sciences\\ 
University of Northern Colorado\\
Greeley, CO 80639}
\email\adzham
\urladdr{\url{www.unco.edu/adzham}}
\thanks{Supported in part by the University of Northern Colorado Summer~2006 SPARC Small Grant
Assistance Program}
\begin{abstract}
We study the Lagrangian properties of the discrete isospectral and isomonodromic 
dynamical systems. We generalize the 
Moser-Veselov approach to integrability of discrete isospectral systems via the 
re-factorization of matrix polynomials to  matrix rational 
functions with a simple divisor and consider in detail the case of two poles or,
equivalently, of two elementary factors. In this case we establish, by  
explicitly writing down the Lagrangian, that the
isospectral dynamic is Lagrangian. Next, we show how to make this Lagrangian time-dependent 
to obtain the equations of the isomonodromic dynamic. In some special cases such 
equations are known to reduce to the difference Painlev\'e equations.  
We show how to obtain the difference Painlev\'e V equation in that way, establishing that 
dPV can be written in the Lagrangian form.
\end{abstract} 
\maketitle

\section{Introduction}\label{sec:intro} 

The theory of completely integrable systems and soliton equations is justly known for its rich and
often unexpected connections with a wide range of other branches of mathematics and mathematical
physics. In recent years its discrete variant, the theory of discrete completely integrable systems,
started to attract a considerable amount of attention. This subject is a part of a more general
field of discrete Lagrangian mechanics, which itself is gaining importance partly
due to the development of new numerical algorithms based on discrete variational integrals, 
see, for example, the recent survey by J.~Marsden and M.~West, \cite{MarWes:2001:DMVI}. 
In a series of papers \cite{Ves:1988:ISWDTDO,MosVes:1991:DVSCISFMP,Ves:1991:ILRFMP} 
A.~Veselov and J.~Moser showed that the discrete analogues 
of many classical integrable systems, e.g., the Neumann system and the spinning top, 
are related to the re-factorization transformation of certain matrix polynomials. Such a 
representation explains the integrability mechanism for these systems, since it 
is a discrete version of the Lax-pair representation. Hence, similarly to the continuous case, 
it can be used to integrate the system using theta functions.
A large number of such examples
can also be found in a recent encyclopedic book by Suris, \cite{Sur:2003:PIDHA}. The 
relationship between matrix factorizations and integrable systems was observed earlier by 
Symes \cite{Sym:1981:ASFNTL}, 
see also a related work by Deift et al, \cite{DeiLiTom:1989:MFIS}. 

The dynamic generated just by the re-factorization transformations is also known
as the \emph{isospectral dynamic}. Combining re-factorization with a shift in the
spectral variable results in a different dynamics called \emph{isomonodromic},
since it originates in the theory of isomonodromic transformations of systems 
of linear difference equations recently developed by A.~Borodin 
\cite{Bor:2004:ITLSDE}, see also \cite{Kri:2004:ATDEWRECRP}. The theory of discrete isomonodromic
transformations is important in part because, similarly to 
the continuous case, under certain conditions 
such transformations give rise to the \emph{discrete Painlev\'e equations} from 
Sakai's hierarchy, \cite{Sak:2001:RSAWARSGPE}, thus clarifying the geometry of these equations.

In \cite{Krichever:uq} I.~Krichever conjectured that both the isospectral and the 
isomonodromic dynamic can (and maybe should) be considered from the Lagrangian point
of view. In this paper we make a first step towards verifying this conjecture. 
We generalize the Moser-Veselov approach from the matrix polynomials to 
a large class of rational matrix functions on the Riemann sphere 
whose determinant divisor is \emph{simple}. This generalization is 
important if one wants to consider discrete integrable 
systems that have higher-genus spectral curves. In addition, rational matrices 
of this type play an important role in 
Krichever's approach, \cite{Kri:2004:ATDEWRECRP}, to the isomonodromic deformations.
Such matrices have natural factorization, where factors correspond to poles
of the determinant divisor. In this paper we focus our attention on the simplest non-trivial case
of two poles (and hence, two factors) and study in detail what happens when two factors are 
interchanged. At this point there is no restrictions on the rank of the matrices.
Our first result is that in this case both the isospectral and the isomonodromic dynamic
is \emph{Lagrangian}. Namely, we introduce a special coordinate system and then
explicitly write down the expressions for the Lagrangian functions. Next, we
restrict the rank to be $2$ and verify that in this case the isomonodromic dynamic 
gives rise to the difference Painlev\'e equation dPV of the Sakai's hierarchy, 
thus establishing that dPV can be written in the Lagrangian form.
This is our second result.

In the remainder of the introduction we give a detailed description of the 
setup of the problem and of our results.

\subsection{Discrete Lagrangian systems}\label{ssec:disc_lagr_sys} 
Continuous dynamical systems can be considered in the Lagrangian
or in the Hamiltonian framework. Of those two,  the \emph{Lagrangian} approach is the one
that naturally generalizes to the discrete case, see 
\cite{Ves:1988:ISWDTDO, MosVes:1991:DVSCISFMP,MarWes:2001:DMVI}.
Let $\mathcal{Q}$ be the 
configuration space of our system and let $n\in\mathbb{Z}$ be the discrete time parameter. 
In the continuous case the Lagrangian 
$\mathcal{L}\in \mathcal{F}(T \mathcal{Q})$ is a function on the tangent bundle of $\mathcal{Q}$. 
For the discrete case we need to change the point $(\mathbf{Q},\dot{\mathbf{Q}})$ in the
tangent space to the pair of points $(\mathbf{Q},\tilde{\mathbf{Q}})$ in the configuration
space itself. The Lagrangian then becomes a function on the square of the
configuration space, $\mathcal{L}\in \mathcal{F}(\mathcal{Q}\times \mathcal{Q})$. The action 
functional $\mathcal{S}$ is then defined on the space of sequences $\{\mathbf{Q}_{k}\}$, 
$k\in\mathbb{Z}$ by the formal sum 
\begin{equation*}
  \mathcal{S}(\{\mathbf{Q}_{k}\}) = \sum_{k} \mathcal{L}(\mathbf{Q}_{k}, \mathbf{Q}_{k+1}),
\end{equation*}  
and the variational principle $\delta \mathcal{S} = 0$ that selects the trajectories of the
system, when written in a coordinate chart, takes the form of the 
\emph{discrete Euler-Lagrange equations}
\begin{equation}
  \frac{\partial \mathcal{L}}{\partial \mathbf{Y}} (\undertilde{\mathbf{Q}},\mathbf{Q}) + 
  \frac{\partial \mathcal{L}}{\partial \mathbf{X}} (\mathbf{Q},\widetilde{\mathbf{Q}}) = 0,
  \label{eq:discrEL}
\end{equation}
where we use the notation $\undertilde{\mathbf{Q}} = \mathbf{Q}_{k-1}$, $\mathbf{Q} = \mathbf{Q}_{k}$,
and $\widetilde{\mathbf{Q}} = \mathbf{Q}_{k+1}$. 
These equations then implicitly define the map (or, more precisely, a correspondence)
$\widetilde{\mathbf{Q}} = \phi(\undertilde{\mathbf{Q}},\mathbf{Q})$, which in turn defines 
the shift (or step) map  $\Phi:\mathcal{Q}\times \mathcal{Q} \to 
\mathcal{Q}\times \mathcal{Q}$ by $\Phi(\undertilde{\mathbf{Q}}, \mathbf{Q}) = 
(\mathbf{Q},\widetilde{\mathbf{Q}})$. The map $\Phi$ is symplectic w.r.t.~the 2-form 
$\sigma = \dfrac{\partial^{2} \mathcal{L}}{\partial \mathbf{X} \partial 
\mathbf{Y}}d\mathbf{X}\wedge d\mathbf{Y}$ on $\mathcal{Q}\times \mathcal{Q}$. Alternatively,
using the discrete version of the Legendre transform by defining the conjugated
momentum
$\mathbf{P} = \dfrac{\partial \mathcal{L}}{\partial 
\mathbf{Y}}(\undertilde{\mathbf{Q}},\mathbf{Q})\in T_{\mathbf{Q}}^{*}\mathcal{Q}$, we see that 
the discrete Euler-Lagrange equations are equivalent to the system
\begin{equation}
\left\{
\begin{aligned}
  \mathbf{P} &=-\dfrac{\partial\mathcal{L}}{\partial
   \mathbf{X}}(\mathbf{Q},\widetilde{\mathbf{Q}})\\[0.1\baselineskip]
  \widetilde{\mathbf{P}} &= \dfrac{\partial \mathcal{L}}{\partial \mathbf{Y}} (\mathbf{Q},
  \widetilde{\mathbf{Q}}),
\end{aligned} 
\right.  \label{eq:discrEL-Ham}
\end{equation}
where the first equation follows from (\ref{eq:discrEL}) and is an implicit equation for 
$\widetilde{\mathbf{Q}}=\widetilde{\mathbf{Q}}(\mathbf{Q},\mathbf{P})$. Hence we get
a map $\Psi:T^{*}_{\mathcal{Q}}\to T^{*}_{\widetilde{\mathcal{Q}}}$,
$\Psi(\mathbf{Q},\mathbf{P}) = (\widetilde{\mathbf{Q}},\widetilde{\mathbf{P}})$ which is
symplectic w.r.t.~the standard symplectic structure. In what follows by the \emph{equations of 
motion} of a discrete Lagrangian system we mean either (\ref{eq:discrEL}) or (\ref{eq:discrEL-Ham}), 
with the corresponding discrete dynamics given by the maps $\Phi$ or $\Psi$ respectively.

\subsection{Discrete integrable systems}\label{ssec:discr_integr_sys} 
The Moser--Veselov approach to the discrete \emph{integrable} systems 
is based on the discrete version of the Lax pair representation and can be briefly described 
as follows. Given a discrete dynamical system, we look for a class $\mathcal{P}$ 
of matrix polynomials $\mathbf{L}(z)$ in 
a spectral variable $z$ and a parameterization map $\eta:\mathcal{Q}\times \mathcal{Q}\to 
\mathcal{P}$, defined on some dense open set, such that:
\begin{enumerate}[(i)]
  \item there is a well-defined factorization rule $\mathbf{L}(z) = \mathbf{L}_{1}(z) 
  \mathbf{L}_{2}(z)$, where the ordering is important, such that the re-factorized matrix
  $\tilde{L}$ obtained by the interchanging the order of the factors,
  $\tilde{\mathbf{L}}(z)=\mathbf{L}_{2}(z)\mathbf{L}_{1}(z)$,
  is again in $\mathcal{P}$ and hence can be written as 
  $\tilde{\mathbf{L}}(z) = \tilde{\mathbf{L}}_{1}(z)
  \tilde{\mathbf{L}}_{2}(z)$, this rule defines the re-factorization map 
  $R:\mathcal{P}\to \mathcal{P}$;
  \item under the parameterization $\eta$ the re-factorization map $R$  corresponds to 
  the shift map $\Phi$ of our discrete dynamical system.
\end{enumerate}
Note that the re-factorization map can also be written in the form 
\begin{equation}
  \tilde{\mathbf{L}} = \mathbf{M}^{-1}\mathbf{L}\mathbf{M},\label{eq:discLaxPair}
\end{equation}
where $\mathbf{M} = \mathbf{L}_{1}(z)$. Equation~(\ref{eq:discLaxPair}) is known as a \emph{discrete
Lax pair} representation of the system. Similarly to the continuous case, finding such a 
representation shows that the dynamic of the system is isospectral. Therefore, it preserves the 
spectral curve $\Gamma$ of the operator $\mathbf{L}(z)$, which implies the integrability 
of the system and also makes it possible to 
obtain the $\theta$-function formulas for solutions of the system in the usual way, see 
\cite{Ves:1991:ILRFMP}.

\subsection{Rational matrices anzats}\label{ssec:rational_matrix_functions_anzats} 
Let us now allow $\mathbf{L}(z)$ to be a meromorphic $r\times r$-matrix
function on the Riemann sphere with the poles $z_{1},\dots,z_{n}$. 
We restrict our attention to the matrices that are generic in the
following sense:
\begin{enumerate}[(i)]
  \item all poles $z_{i}$ of $\mathbf{L}(z)$ are simple;
  \item the divisor of  $\mathbf{L}(z)$ is \emph{simple} as well, where by the
  \emph{divisor of $\mathbf{L}(z)$} we mean the divisor of its 
   determinant function,
  $\mathcal{D} = (\mathbf{L}(z))=(\det\mathbf{L}(z)) = \sum_{i}z_{i} -
  \sum_{j}\zeta_{j}$. This is equivalent to the condition that the residue matrices
  $\mathbf{L}_{k} := \operatorname{res}_{z_{k}}\mathbf{L}(z)$ are of
  \emph{rank one}.
\end{enumerate}
For a fixed divisor $\mathcal{D}$ we denote the space of all such matrices by $\mathcal{M}_{r}^{D}$.

Without any loss of generality we can further restrict out attention to the case  
$z_{0}=\infty\notin \mathcal{D}$ and
$\mathbf{L}_{0}=\lim_{z\to\infty}\mathbf{L}(z)$ is invertible and diagonalizable. 
Any such matrix $\mathbf{L}(z)$ has two different representations, \emph{additive}:
\begin{align}
  \mathbf{L}(z) &= \mathbf{L}_{0}  + \sum_{k}\frac{\mathbf{L}_{k}}{z-z_{k}},
  \label{eq:L-addrep}\\
  \intertext{and \emph{multiplicative}:}
  \mathbf{L}(z) &= \prod_{k}\left(\mathbf{A} + \frac{\mathbf{G}_{k}}{z-z_{k}} \right),
  \label{eq:L-mult}
\end{align}
  where $\mathbf{G}_{k}$ is a matrix of rank one and $\mathbf{A}^{k} = \mathbf{L}_{0}$; we 
  are mainly interested in the multiplicative representation. We call the factors
  $\displaystyle\mathbf{B}_{i}^{\mathbf{A}}(z) = 
  \left( \mathbf{A} + \frac{\mathbf{G}_{i}}{z-z_{i}}\right)$ in the 
  multiplicative representation the \emph{elementary divisors} of $\mathbf{L}(z)$. Note that 
  the ordering of the poles determines the ordering of the factors in the multiplicative
  representation, which is what we need to define the re-factorization map.
  
  In this paper we restrict our attention to the two-pole case. However, since any permutation
  is a composition of elementary transposition, any re-factorization transformation
  is generated by a sequence of transformations that we consider below. Thus, we expect our results to 
  hold in the general case as well, but this question will be considered elsewhere.  

\subsection{The re-factorization transformation and the isospectral discrete dynamical 
system}\label{ssec:the_configuration_space} 
Let us now fix the divisor $\mathcal{D} = z_{1} + z_{2} - 
\zeta_{1} - \zeta_{2}$, where all four points are finite and distinct, and consider the 
re-factorization map $R: \mathcal{M}_{r}^{D}\to \mathcal{M}_{r}^{D}$ given by 
\begin{equation*}
\mathbf{L}(z) = \mathbf{B}_{1}^{\mathbf{A}}(z) \mathbf{B}_{2}^{\mathbf{A}}(z) \mapsto
\tilde{\mathbf{L}}(z) = \mathbf{B}_{2}^{\mathbf{A}}(z) \mathbf{B}_{1}^{\mathbf{A}}(z)  
= \tilde{\mathbf{B}}_{1}^{\mathbf{A}}(z) \tilde{\mathbf{B}}_{2}^{\mathbf{A}}(z).
\end{equation*}
We want to determine whether there is a natural discrete
dynamical system for which this map is a discrete Lax pair representation, and if so, what 
is the Lagrangian of this system. Note that this setting is rather general, since no 
restrictions on the rank $r$ are imposed. First it is necessary to identify 
a configuration space $\mathcal{Q}$ such that there is a parameterization map
$\eta: \mathcal{Q}\times \mathcal{Q}\to \mathcal{M}_{r}^{\mathcal{D}}$ satisfying the following
diagram:
\begin{equation*}
\xymatrix{(\undertilde{\mathbf{Q}},\mathbf{Q})\ar[d]_{\eta} \ar[r]^{\Phi} & 
  (\mathbf{Q},\tilde{\mathbf{Q}})\ar[d]^{\eta} \\ 
  \mathbf{L}(z) = \undertilde{\mathbf{B}}_{2}^{\mathbf{A}}(z) 
  \undertilde{\mathbf{B}}_{1}^{\mathbf{A}}(z) = \mathbf{B}_{1}^{\mathbf{A}}(z)
  \mathbf{B}_{2}^{\mathbf{A}}(z) \ar[r]^{R}& 
  \tilde{\mathbf{L}}(z) = \mathbf{B}_{2}^{\mathbf{A}}(z)
  \mathbf{B}_{1}^{\mathbf{A}}(z) = \tilde{\mathbf{B}}_{1}^{\mathbf{A}}(z) 
  \tilde{\mathbf{B}}_{2}^{\mathbf{A}}(z)\ .}
\end{equation*}
From this diagram it is clear that half the data in $\mathbf{L}(z)$ should come from 
$\undertilde{\mathbf{\mathbf{Q}}}$ and half should come from $\mathbf{Q}$. Moreover, 
this data should be of the same type to be compatible with the shift map $\Phi$. We know that 
the elementary divisors $\mathbf{B}^{\mathbf{A}}_{i}(z)$ completely determine $\mathbf{L}(z)$, 
and each elementary divisor is in turn determined by the rank-one matrix $\mathbf{G}_{i} = 
\mathbf{p}_{i}\mathbf{q}^{\dag}_{i}$, where $\mathbf{p}_{i}$ and $\mathbf{q}^{\dag}_{i}$ are
defined up to a common scaling constant, and this constant can be recovered from the divisor 
$\mathcal{D}$. Thus, $\mathbf{L}(z)$ is completely determined by either 
$\undertilde{\mathbf{p}}_{1}$, 
$\undertilde{\mathbf{p}}_{2}$, $\undertilde{\mathbf{q}}^{\dag}_{1}$, 
$\undertilde{\mathbf{q}}^{\dag}_{2}$ or $\mathbf{p}_{1}$, 
$\mathbf{p}_{2}$, $\mathbf{q}^{\dag}_{1}$, $\mathbf{q}^{\dag}_{2}$. In view of that
 we take $\undertilde{\mathbf{Q}}=(\undertilde{\mathbf{p}}_{1}, 
\undertilde{\mathbf{q}}^{\dag}_{2})$, $\mathbf{Q} = (\mathbf{p}_{1},\mathbf{q}^{\dag}_{2})$.
A priori $\mathbf{Q}\in\mathbb{C}^{r}\times (\mathbb{C}^{r})^{\dag}$, but in fact the resulting
expressions are homogeneous in $\mathbf{p}_{i}$, $\mathbf{q}^{\dag}_{i}$, 
and so the correct configuration space is 
$\mathcal{Q}=\mathbb{P}^{r-1}\times (\mathbb{P}^{r-1})^{\dag}$. 
In Theorem~\ref{thm:EQ-M} we give an explicit description of the
parameterization map $\eta: \mathcal{Q}\times \mathcal{Q} \to \mathcal{M}^{\mathcal{D}}_{r}$,
compute the corresponding equations of motion, and show that these
equations of motion are the discrete Euler-Lagrange equations with the Lagrangian function 
  $\mathcal{L}$ given by
  \begin{align*}
    \mathcal{L}(\mathbf{X},\mathbf{Y}) &= 
    (z_{2} - z_{1}) \log(\mathbf{x}^{\dag}_{2} \mathbf{x}_{1}) + 
    (z_{1} - \zeta_{2}) \log(\mathbf{x}^{\dag}_{2} \mathbf{A}^{-1} \mathbf{y}_{1}) + \\ &\qquad
    (\zeta_{2} - \zeta_{1})\log(\mathbf{y}^{\dag}_{2} \mathbf{A}^{-2} \mathbf{y}_{1}) + 
    (\zeta_{1} - z_{2}) \log(\mathbf{y}^{\dag}_{2}\mathbf{A}^{-1} \mathbf{x}_{1}).
  \end{align*}

\subsection{Discrete Painlev\'e equations and the isomonodromic 
transformations of the systems of linear difference equations}
\label{ssec:discrete_painlev_e_equations} 
Another natural discrete dynamics that can be considered on our 
space of matrices is the \emph{isomonodromic} discrete dynamical system. One of the reasons this
system is interesting is its relationship to the theory of the 
discrete Painlev\'e equations.

The recent surge of interest in the discrete version of the 
famous Painlev\'e equations is in part 
due to the fact that these equations appear in the calculation of \emph{discrete gap
probabilities} in the theory of \emph{(determinantal) Random Point
Processes}, \cite{Bor:2003:DPDPE, BorDei:2002:FDJTRT}. In addition, 
H.~Sakai in \cite{Sak:2001:RSAWARSGPE} described a very 
elegant and purely geometric approach to the discrete 
Painlev\'e equations using the Cremona action on the algebraic surfaces. More information 
about the current progress in the theory of discrete Painlev\'e equations can be found  in
\cite{GraRam:2004:DPER}. 

In the continuous case there is a well-known relationship between the 
isomonodromic transformations of the flat meromorphic connections on the 
Riemann sphere and the Painlev\'e equations. Thus, it is natural to expect that the
discrete Painlev\'e equations should be related to the
isomonodromic deformations of \emph{matrix linear difference
  equations}. However, there is a serious obstacle --- the notion of
monodromy for a \emph{differential} equation has \emph{no obvious
  generalization} to a \emph{difference} equation, and only recently some
  significant progress was made in this direction. Recall that the general theory of
matrix linear differential equations
\begin{equation*}
  \label{eq:dEQ}
  \mathbf{\Psi}(z+1) = \mathbf{L}(z)\mathbf{\Psi}(z)
\end{equation*}
goes back the works of George Birkhoff, \cite{Bir:1911:GTLDE}. First step in 
Birkhoff's approach was to use a special gauge transformation
to clear all poles of $\mathbf{L}(z)$ and make it a
\emph{polynomial} in $z$. Note that as a result we get a \emph{pole of higher order at
infinity}. Next, Birkhoff showed that there are two canonical
meromorphic solutions $\mathbf{\Psi}_{l}(z)$ and
$\mathbf{\Psi}_{r}(z)$ that have the prescribed asymptotic
behavior for $\Re(z)\ll 0$ and $\Re(z)\gg0$ respectively.
Then the analogue of the monodromy map is just the \emph{connection
  matrix} $\mathbf{C}$ of these solutions, $\mathbf{C}(z) =
\mathbf{\Psi}_{r}^{-1}(z)\mathbf{\Psi}_{l}(z)$. Birkhoff
also showed that in this situation the \emph{isomonodromic transformations}
$\mathbf{L}(z)\mapsto \tilde{\mathbf{L}}(z)$ that preserve
$\mathbf{C}(z)$  are given by 
\begin{equation}
  \label{eq:dIsom}
  \tilde{\mathbf{L}}(z) =
\mathbf{R}(z+1) \mathbf{L}(z)\mathbf{R}^{-1}(z),
\end{equation}
where
$\mathbf{R}(z)$ is a \emph{rational} matrix. 
A.~Borodin, in \cite{Bor:2004:ITLSDE}, constructed a general theory of
such transformations for polynomial $\mathbf{L}(z)$ and showed that it give rise to the difference
\emph{Schlesinger equations}. These equations, when the 
the space of parameters is two-dimensional, can in turn be reduced to the
difference Painlev\'e equations. In a follow-up paper
\cite{AriBor:2006:MSDDPE}, D.~Arinkin and A.~Borodin showed, using a more
geometric language of $d$-connections,
that for some special cases Sakai's surfaces can be
identified with the moduli space of such $d$-connections and that the isomonodromic
transformations can then be though of as the \emph{elementary modifications} of
$d$-connections, which are in turn given by the difference Painlev\'e
equations (examples considered in this paper are dPV and
dPVI). This result helps to explain the geometry behind the
difference \emph{isomonodromy}--\emph{Painlev\'e} correspondence.
Recently, Arinkin and Borodin found the description for the $\tau$-function 
of the discrete isomonodromy transformations for both polynomial and rational cases, 
see \cite{AriBor:2007:TDITP}.

A different approach to the notion of the monodromy of a
linear \emph{difference} equation  was  suggested by
I.~Krichever in \cite{Kri:2004:ATDEWRECRP}. In this approach the matrix
$\mathbf{L}(z)$ belongs to the same anzats as we consider in the present paper --- 
it is \emph{regular} at infinity, all of its
poles $z_{i}$ are \emph{finite and simple}, and 
$\operatorname{res}_{z_{i}}\mathbf{L}(z)$ are of rank
\emph{one}. For such matrices Krichever introduced the notion of a
\emph{local monodromy} that can be thought of as a monodromy
corresponding to the path around a pole, and also constructed the
isomonodromy transformations, that again have the form~(\ref{eq:dIsom}). 
He also showed how to generalize this theory
from rational to \emph{elliptic} functions.

In the present paper we consider a special case
of the transformation~(\ref{eq:dIsom}) that has the form
\begin{equation}
  \mathbf{L}(z)=\mathbf{B}^{\mathbf{A}}_{1}(z) \mathbf{B}^{\mathbf{A}}_{2}(z) \mapsto
  \tilde{\mathbf{L}}(z) = \tilde{\mathbf{B}}^{\mathbf{A}}_{1}(z) 
  \tilde{\mathbf{B}}^{\mathbf{A}}_{2}(z) 
  =\mathbf{B}^{\mathbf{A}}_{2}(z+1) \mathbf{B}^{\mathbf{A}}_{1}(z) = 
  \mathbf{B}^{\mathbf{A}}_{2}(z+1) \mathbf{L}(z) (\mathbf{B}^{\mathbf{A}}_{2}(z))^{-1},
    \label{eq:dIsomSp}
\end{equation}
where $\mathbf{B}^{\mathbf{A}}_{i}(z)$ are the elementary divisors defined earlier.
We show that, similarly to the isospectral case, these transformations can be written in the 
Lagrangian form. The main new feature of the isomonodromic approach is the fact that
such transformation changes the divisor $D$ to the divisor $\tilde{D}$, where 
$\tilde{z}_{1}=z_{1}$, $\tilde{\zeta}_{1} = \zeta_{1}$, $\tilde{z}_{2} = z_{2} - 1$, and
$\tilde{\zeta}_{2} = \zeta_{2} - 1$. Thus we need to make the Lagrangian $\mathcal{L}$
time-dependent by putting $z_{2}(t) = z_{2} - t$ and $\zeta_{2}(t) = \zeta_{2} - t$: 
\begin{align*}
  \mathcal{L}(\mathbf{X},\mathbf{Y},t) &= 
  (z_{2}(t) - z_{1}) \log(\mathbf{x}^{\dag}_{2} \mathbf{x}_{1}) + 
  (z_{1} - \zeta_{2}(t)) \log(\mathbf{x}^{\dag}_{2} \mathbf{A}^{-1} \mathbf{y}_{1}) + \\ &\qquad
  (\zeta_{2}(t) - \zeta_{1})\log(\mathbf{y}^{\dag}_{2} \mathbf{A}^{-2} \mathbf{y}_{1}) + 
  (\zeta_{1} - z_{2}(t)) \log(\mathbf{y}^{\dag}_{2}\mathbf{A}^{-1} \mathbf{x}_{1}).
\end{align*}
The time-dependent discrete Euler-Lagrange equations 
\begin{equation*}
  \frac{\partial \mathcal{L}}{\partial \mathbf{Y}} (\mathbf{Q}_{k-1},\mathbf{Q}_{k},k-1) + 
  \frac{\partial \mathcal{L}}{\partial \mathbf{X}} (\mathbf{Q}_{k},\mathbf{Q}_{k+1},k) = 0
  \label{eq:discrEL-time}
\end{equation*}
then describe the isomonodromic dynamics~(\ref{eq:dIsomSp}). Finally, we verify, 
essentially following \cite{AriBor:2006:MSDDPE}, that in 
the rank-two case equation~(\ref{eq:dIsomSp}),
when written in the so-called \emph{spectral coordinates}, reduces to the difference 
Painlev\'e equation dPV of the Sakai's hierarchy, thus establishing that this equation 
can be written in the Lagrangian form.


\subsection{Organization of the paper}\label{ssec:organization_of_the_papaer} 
In Section~\ref{sec:el_divs} we study properties the elementary divisors, and 
obtain the  description of the re-factorization map. 
In Section~\ref{sec:isosp_dyn} we establish the Lagrangian structure of the isospectral dynamics, 
and in Section~\ref{sec:isom_dyn} we extend this result to the isomonodromic case. 


\section{Elementary Divisors}\label{sec:el_divs} 
In representing rational matrix functions in the multiplicative form we take each factor to be
a matrix of the following simple type.
\begin{definition}\label{def:el_div}
  An \emph{elementary divisor} with the simple pole at $z_{i}$
  is a matrix of the form $\mathbf{B}_{i}^{\mathbf{A}}(z) = \mathbf{A} + 
  \dfrac{\mathbf{G}_{i}}{z-z_{i}}$, where $\mathbf{G}_{i}$ is a matrix of rank one and 
   $\mathbf{A}$ is some \emph{fixed} constant non-degenerate matrix (which is usually taken
   to be diagonal). %
\end{definition}

In this section we describe certain useful properties of elementary divisors, and also explain
our normalization conventions.

\subsection{Rank-one matrices and normalization}\label{ssec:rank_one} 
Let us first make some remarks about matrices of rank one. Any such matrix has the form 
$\mathbf{G} = \mathbf{p}\mathbf{q}^{\dag}$ for some column vector $\mathbf{p}$ and some
row vector $\mathbf{q}^{\dag}$, where the vectors $\mathbf{p}$ and $\mathbf{q}^{\dag}$  
are defined up to a common scaling constant. To explicitly keep track of such scaling 
constants during computations we need to normalize these vectors in some way. 

\begin{notation}
  Given the choice of a normalization, we denote by $[\mathbf{v}]$ the normalization
  of a vector $\mathbf{v}$ and by $\nu(\mathbf{v})$ its \emph{normalization constant} 
  w.r.t. this normalization. Thus, 
  $\mathbf{v} = \nu(\mathbf{v})[\mathbf{v}]$. We also use the notation $[\mathbf{v}]$ 
  for the normalized vectors. Hence, any matrix of rank one can be written as 
  $\mathbf{G}=\mathbf{p}\mathbf{q}^{\dag} = \nu(\mathbf{G})[\mathbf{p}] [\mathbf{q}^{\dag}]
   = \lambda [\mathbf{p}][\mathbf{q}^{\dag}]$,
  where $\lambda = \nu(\mathbf{G}) = \nu(\mathbf{p})\nu(\mathbf{q}^{\dag})$ is  the 
  \emph{normalization constant} for $\mathbf{G}$.
\end{notation}

For our purposes it is most convenient to work with \emph{linear normalizations}. 
Such normalizations have the property that any linear relation among the normalized
vectors implies the same linear relation for the coefficients; 
if $a [\mathbf{u}] = \sum_{k} b_{k} [\mathbf{v}_{k}]$, then
$a = \sum_{k} b_{k}$. For example, the normalizations
\begin{enumerate}[(a)]
\item $\sum_{i}(\mathbf{p})^{i}=\sum_{j}(\mathbf{q}^{\dag})_{j}=1$,
\item $(\mathbf{p})^{i}=(\mathbf{q}^{\dag})_{j}=1$ for some choice of
  indexes $i$ and $j$
\end{enumerate}
satisfy this requirement. 

For an elementary divisor we can use one of the following three natural 
normalizations: 
 \begin{equation*}
   \mathbf{B}_{i}^{\mathbf{A}}(z) = \mathbf{A} + \frac{\lambda_{i}[\mathbf{p}_i] 
   [\mathbf{q}^{\dag}_i]}{z-z_i} = \mathbf{A} \left( \mathbf{1} +
   \frac{\lambda_i^{\mathbf{p}} [\mathbf{A}^{-1} \mathbf{p}_i] 
   [\mathbf{q}^{\dag}_i]}{z-z_{i}}\right) = \left( \mathbf{1} + 
   \frac{\lambda_{i}^{\mathbf{q}^{\dag}} 
   [\mathbf{p}_{i}][\mathbf{q}^{\dag}_{i} \mathbf{A}^{-1}]}{z-z_{i}}\right) \mathbf{A},
 \end{equation*}
   where the superscript $\mathbf{p}$ in $\lambda_{i}^{\mathbf{p}}$ indicates that instead 
   of normalizing $\mathbf{p}$ we normalize $\mathbf{A}^{-1}\mathbf{p}$; similarly,
   $\lambda_{i}^{\mathbf{q}^{\dag}}$ corresponds to the normalization of $\mathbf{q}^{\dag}_i 
   \mathbf{A}^{-1}$.

\subsection{Properties of elementary divisors}\label{ssec:prop_el_divs} 
The following Lemma is a key technical tool for working with elementary divisors.

\begin{lemma}\label{lem:el_divs}
Let $\mathbf{B}_{i}^\mathbf{A}(z) = \mathbf{A} + \dfrac{\mathbf{G}_{i}}{z-z_{i}}$ 
and define $\zeta_{i}$ by
the equation $\operatorname{tr}(\mathbf{G}_{i} \mathbf{A}^{-1} ) = z_{i} - \zeta_{i}$. 
Then the following holds.
\begin{enumerate}[(i)]
  \item $\det \mathbf{B}_{i}^{\mathbf{A}}(z) = \dfrac{z-\zeta_{i}}{z-z_{i}} \det{\mathbf{A}} $ and 
   $(\mathbf{B}_{i}^{\mathbf{A}}(z))^{-1} = 
    \mathbf{A^{-1}}\left( \mathbf{A} - 
     \dfrac{\mathbf{G}_{i}}{z-\zeta_{i}}\right)\mathbf{A}^{-1}$.
  \item Knowing how $\mathbf{B}_{i}^{\mathbf{A}}(z)$ operates on row 
    (resp.~column) vectors and also knowing the column (resp.~row) 
    vector of the rank one part allows us to determine 
    $\mathbf{B}_{i}^{\mathbf{A}}(z)$:
     \begin{itemize}
     \item if $\mathbf{v} = \mathbf{B}_{i}^{\mathbf{A}}(z) \mathbf{w} $, then 
       $\mathbf{G}_{i} =
       \mathbf{A} \left((z_{i} - z)\dfrac{\mathbf{w}
       }{\mathbf{q}^{\dag}_{i} \mathbf{w}} + (z - \zeta_{i}) 
       \dfrac{\mathbf{A}^{-1}\mathbf{v}
       }{\mathbf{q}^{\dag}_{i} \mathbf{A}^{-1} \mathbf{v}}\right)  
       \mathbf{q}^{\dag}_{i}$;
     \item if $\mathbf{v}^{\dag} = \mathbf{w}^{\dag}
     \mathbf{B}_{i}^{\mathbf{A}}(z)$, then
       $\mathbf{G}_{i} = 
       \mathbf{p}_{i} \left((z_{i} - z)\dfrac{\mathbf{w}^{\dag}_{i}
       }{\mathbf{w}^{\dag} \mathbf{p}_{i}} + (z - \zeta_{i}) 
       \dfrac{\mathbf{v}^{\dag} \mathbf{A}^{-1}
       }{\mathbf{v}^{\dag}\mathbf{A}^{-1}\mathbf{p}_{i}}\right) \mathbf{A}$.
    \end{itemize}
\end{enumerate}
\end{lemma}
\begin{proof} 
  To prove part (i), note that 
  \begin{equation*}
    \det \mathbf{B}_{i}^{\mathbf{A}}(z)  = \det \left( \mathbf{1} + 
    \frac{\mathbf{G}_{i} \mathbf{A}^{-1}}{(z-z_{i})} 
    \right) \det{\mathbf{A}} = \left( 1 + 
    \frac{\operatorname{tr}(\mathbf{G}_{i}\mathbf{A}^{-1})}{(z-z_{i})}\right) \det\mathbf{A} = 
    \frac{z-\zeta_{i}}{z-z_{i}} \det{\mathbf{A}},
  \end{equation*}
  since $\mathbf{G}_{i} \mathbf{A}^{-1}$ is a matrix of rank one. The formula for the inverse 
  matrix can be checked by the direct calculation.
  
  To establish part (ii) we normalize the elementary divisor. Then, 
  using the linearity property of the normalization, the equation 
    $\mathbf{v} = \mathbf{B}_{i}^{\mathbf{A}}(z) \mathbf{w}$ can be written as
    \begin{align*}
      [\mathbf{A}^{-1}\mathbf{v}] &= \left[ [\mathbf{w}] + \frac{\lambda_{i}^{\mathbf{p}} 
      [\mathbf{A}^{-1} \mathbf{p}_{i}] [\mathbf{q}^{\dag}_{i}][\mathbf{w}]}{z-z_{i}} \right] 
      = \frac{[\mathbf{w}](z-z_{i}) + \lambda_{i}^{\mathbf{p}} [\mathbf{A}^{-1} \mathbf{p}_{i}]
      [\mathbf{q}^{\dag}_{i}] [\mathbf{w}]}{(z-z_{i}) + \lambda_{i}^\mathbf{p} 
      [\mathbf{q}^{\dag}_i][\mathbf{w}]}\\
      \intertext{Multiplying both sides by $[\mathbf{q}^{\dag}_{i}]$ gives 
       $(z-z_{i}) + \lambda_{i}^\mathbf{p} [\mathbf{q}^{\dag}_i][\mathbf{w}] 
       = \dfrac{[\mathbf{q}^{\dag}_{i}][\mathbf{w}] (z-\zeta_{i})}{
       [\mathbf{q}^{\dag}_i][\mathbf{A}^{-1} \mathbf{v}]}$, and so}
      \mathbf{G}_{i} &= 
      \mathbf{A} \left( 
      \lambda_{i}^{\mathbf{p}} [\mathbf{A}^{-1}\mathbf{p}_{i}][\mathbf{q}^{\dag}_{i}] \right) = 
      \mathbf{A} \left(\frac{(z_{i} - z) [\mathbf{w}] [\mathbf{q}^{\dag}_{i}]}{ [\mathbf{w}] 
      [\mathbf{q}^{\dag}_{i}]} + 
      \frac{(z-\zeta_{i})[\mathbf{A}^{-1} \mathbf{v}] [\mathbf{q}^{\dag}_{i}]}{
      [\mathbf{q}^{\dag}_{i}] [\mathbf{A}^{-1} \mathbf{v}]}\right).
    \end{align*}
    Since the expression in the parentheses is homogeneous, we can remove the normalization 
    brackets to get the desired result. Second formula is obtained in a similar way.
\end{proof}

\subsection{The re-factorization transformation}\label{sub:re_fact} 

Consider now the following question. Let
\begin{equation*}
  \mathbf{B}_{2}^{\mathbf{A}}(z) \mathbf{B}_{1}^{\mathbf{A}}(z) = 
  \widetilde{\mathbf{B}}_{1}^{\mathbf{A}}(z) \widetilde{\mathbf{B}}_{2}^{\mathbf{A}}(z),
\label{eq:full-flip}
\end{equation*}
where $\mathbf{B}_{i}^{\mathbf{A}}(z)$ and $\widetilde{\mathbf{B}}_{i}^{\mathbf{A}}(z)$ 
are elementary divisors with the simple poles at $z_{i}$. What are the 
relationships between the vectors that form their rank-one parts? 
To begin with, note that taking the determinant results in the
 equation $\displaystyle\frac{(z-\zeta_{i})(z-\zeta_{2})}{(z-z_{i})(z-z_{2})} 
 = \frac{(z-\tilde{\zeta_{1}})(z-\tilde{\zeta_{2}})}{(z-\tilde{z_{1}})(z-\tilde{z}_{2})}$, where
 $\tilde{z_{i}} = z_{i}$ by definition, and so we must have either $\tilde{\zeta_{i}} = 
 \zeta_{i}$ (the general case) or $\tilde{\zeta_{1}} = \zeta_{2}$ and $\tilde{\zeta_{2}} = 
 \zeta_{1}$ (which is a special case, since it requires a non-trivial relationship between
 the poles and the rank-one parts of the elementary divisors,
 \begin{equation}
  z_{2} - z_{1} = \operatorname{tr}(\tilde{\mathbf{G}}_{2}\mathbf{A}^{-1}) - 
  \operatorname{tr}(\mathbf{G}_{1}\mathbf{A}^{-1}) = 
  \operatorname{tr}(\mathbf{G}_{2}\mathbf{A}^{-1}) - 
  \operatorname{tr}(\tilde{\mathbf{G}}_{1}\mathbf{A}^{-1}).
 \end{equation}
  The Theorem below explains the general case, and the special case can be considered 
 in  exactly the same way.

\begin{theorem}\label{thm:flip}
  Let $\mathbf{B}_{2}^{\mathbf{A}}(z) \mathbf{B}_{1}^{\mathbf{A}}(z) = 
  \widetilde{\mathbf{B}}_{1}^{\mathbf{A}}(z) \widetilde{\mathbf{B}}_{2}^{\mathbf{A}}(z)$
  and $\tilde{\zeta}_{i} = \zeta_{i}$. Then the following holds.
  \begin{enumerate}[(i)]
    \item The vectors $\mathbf{p}_{i}$, $\tilde{\mathbf{p}}_{i}$, $\mathbf{q}^{\dag}_{i}$,
    $\tilde{\mathbf{q}}^{\dag}_{i}$ are related by 
    \begin{alignat*}{4}
      [\tilde{\mathbf{p}}_{1}] &= \big[\mathbf{B}_{2}^{\mathbf{A}}(z_{1}) \mathbf{p}_{1}\big] & 
      &= \big[ \mathbf{A} \widetilde{\mathbf{B}}_{2}^{\mathbf{A}}(\zeta_{1}) \mathbf{A}^{-1} 
      \mathbf{p}_{1}\big]  \qquad & \qquad
      [\mathbf{p}_{2}] &= \big[ \widetilde{\mathbf{B}}_{1}^{\mathbf{A}}(z_{2}) 
      \tilde{\mathbf{p}}_{2}\big] &
      &= \big[ \mathbf{A} \mathbf{B}_{1}^{\mathbf{A}}(\zeta_{2})
      \mathbf{A}^{-1}\tilde{\mathbf{p}}_{2} \big] \\
      [\mathbf{q}^{\dag}_{1}] &= \big[\tilde{\mathbf{q}}^{\dag}_{1} 
      \widetilde{\mathbf{B}}_{2}^{\mathbf{A}}(z_{1})\big] & 
      &= \big[\tilde{\mathbf{q}}^{\dag}_{1} \mathbf{A}^{-1} 
      \mathbf{B}_{2}^{\mathbf{A}}(\zeta_{1}) \mathbf{A}\big] 
      \qquad & \qquad
      [\tilde{\mathbf{q}}^{\dag}_{2}] &= \big[\mathbf{q}^{\dag}_{2} 
      \mathbf{B}_{1}^{\mathbf{A}}(z_{2})\big] &
      &= \big[\mathbf{q}^{\dag}_{2} \mathbf{A}^{-1} 
      \widetilde{\mathbf{B}}_{1}^{\mathbf{A}}(\zeta_{2}) \mathbf{A}\big]. 
    \end{alignat*}
    \item The vectors $\mathbf{p}_{1}$, $\tilde{\mathbf{p}}_{1}$, $\mathbf{q}^{\dag}_{2}$,
    $\tilde{\mathbf{q}}^{\dag}_{2}$ completely determine the elementary divisors 
    $\mathbf{B}_{i}^{\mathbf{A}}(z)$, $\tilde{\mathbf{B}}_{i}^{\mathbf{A}}(z)$ via
    \begin{align}
      \mathbf{G}_{1}  &= \left( (z_{1} - z_{2}) 
       \frac{\mathbf{p}_{1}\mathbf{q}^{\dag}_{2}}{\mathbf{q}^{\dag}_{2} \mathbf{p}_{1}} + 
       (z_{2} - \zeta_{1}) \frac{\mathbf{p}_{1}\tilde{\mathbf{q}}^{\dag}_{2} \mathbf{A}^{-1}
       }{\tilde{\mathbf{q}}^{\dag}_{2}\mathbf{A}^{-1}\mathbf{p}_{1}}\right)\mathbf{A},
       \label{eq:q1_ptp1qtq2} \\
       \tilde{\mathbf{G}}_{1} &= \left(
       (z_{1} - \zeta_{2}) \frac{\tilde{\mathbf{p}}_{1}\mathbf{q}^{\dag}_{2} \mathbf{A}^{-1}
       }{\mathbf{q}^{\dag}_{2} \mathbf{A}^{-1} \tilde{\mathbf{p}}_{1}} + 
       (\zeta_{2} - \zeta_{1}) \frac{\tilde{\mathbf{p}}_{1} \tilde{\mathbf{q}}^{\dag}_{2}
        \mathbf{A}^{-2}
       }{\tilde{\mathbf{q}}^{\dag}_{2} \mathbf{A}^{-2} \tilde{\mathbf{p}}_{1}}\right)\mathbf{A}.
       \label{eq:tq1_ptp1qtq2} \\
      \mathbf{G}_{2}  &= \mathbf{A} \left( (z_{2}- z_{1}) 
      \frac{\mathbf{p}_{1} \mathbf{q}^{\dag}_{2}}{\mathbf{q}^{\dag}_{2} \mathbf{p}_{1}} + 
      (z_{1} - \zeta_{2}) \frac{\mathbf{A}^{-1} \tilde{\mathbf{p}}_{1} \mathbf{q}^{\dag}_{2}
      }{\mathbf{q}^{\dag}_{2} \mathbf{A}^{-1} \tilde{\mathbf{p}}_{1}}\right), 
      \label{eq:p2_ptp1qtq2}\\
      \tilde{\mathbf{G}}_{2}  &= \mathbf{A} \left( 
      (z_{2} - \zeta_{1}) \frac{\mathbf{A}^{-1}\mathbf{p}_{1} \tilde{\mathbf{q}}^{\dag}_{2}
      }{\tilde{\mathbf{q}}^{\dag}_{2} \mathbf{A}^{-1} \mathbf{p}_{1}} + 
      (\zeta_{1} - \zeta_{2}) \frac{\mathbf{A}^{-2} \tilde{\mathbf{p}}_{1} 
      \tilde{\mathbf{q}}^{\dag}_{2}
      }{\tilde{\mathbf{q}}^{\dag}_{2} \mathbf{A}^{-2} \tilde{\mathbf{p}}_{1}}\right),%
      \label{eq:tp2_tpt1qtq2}
     \end{align}
  \end{enumerate}
\end{theorem}

\begin{proof}
  The equations in part (i) are 
   obtained by taking the residues of the equation $\mathbf{B}_{2}^{\mathbf{A}}(z)     
    \mathbf{B}_{1}^{\mathbf{A}}(z) = 
   \widetilde{\mathbf{B}}_{1}^{\mathbf{A}}(z) \widetilde{\mathbf{B}}_{2}^{\mathbf{A}}(z)$ 
   at the points $z_{i}$ and the residues of the inverse equation at the points $\zeta_{i}$ 
   and then equating the normalized 
   row and column vectors of the resulting rank-one matrices. 
   Using Lemma~\ref{lem:el_divs} we then obtain the equations in part (ii).
\end{proof}



\section{The Isospectral Case}\label{sec:isosp_dyn} 

We are now in the position to describe the equations of motion for the isospectral dynamics.

\begin{theorem}\label{thm:EQ-M} 
  Let $\undertilde{\mathbf{Q}} = (\undertilde{\mathbf{p}}_{1},
  \undertilde{\mathbf{q}}^{\dag}_{2})$, $\mathbf{Q} = (\mathbf{p}_{1}, \mathbf{q}^{\dag}_{2})$, 
  $\tilde{\mathbf{Q}} = (\tilde{\mathbf{p}}_{1}, \tilde{\mathbf{q}}^{\dag}_{2})$,
  where each vector is considered modulo re-scaling, 
  \begin{align*}
    \eta(\undertilde{\mathbf{Q}},\mathbf{Q}) &= 
    \mathbf{L}(z) = \undertilde{\mathbf{B}}_{2}^{\mathbf{A}}(z) 
    \undertilde{\mathbf{B}}_{1}^{\mathbf{A}}(z) = \mathbf{B}_{1}^{\mathbf{A}}(z) 
    \mathbf{B}_{2}^{\mathbf{A}}(z),
    \intertext{and}
    \eta(\mathbf{Q},\tilde{\mathbf{Q}}) &= \tilde{\mathbf{L}}(z) = 
    \mathbf{B}_{2}^{\mathbf{A}}(z)
    \mathbf{B}_{1}^{\mathbf{A}}(z) = \tilde{\mathbf{B}}_{1}^{\mathbf{A}}(z) 
    \tilde{\mathbf{B}}_{2}^{\mathbf{A}}(z).
  \end{align*}
  Then
  \begin{enumerate}[(i)]
    \item The parameterization $\eta(\undertilde{\mathbf{Q}},\mathbf{Q})$ is given by 
    \begin{align*}
      \eta(\undertilde{\mathbf{Q}},\mathbf{Q}) = \mathbf{L}(z) &= 
      \left(\mathbf{A} + \frac{1}{z-z_{1}}\left(
       (z_{1} - \zeta_{2}) \frac{\mathbf{p}_{1}\undertilde{\mathbf{q}}^{\dag}_{2} 
       }{\undertilde{\mathbf{q}}^{\dag}_{2} \mathbf{A}^{-1} \mathbf{p}}_{1} + 
       (\zeta_{2} - \zeta_{1}) \frac{\mathbf{p}_{1} \mathbf{q}^{\dag}_{2}
        \mathbf{A}^{-1}
       }{\mathbf{q}^{\dag}_{2} \mathbf{A}^{-2} \mathbf{p}_{1}}\right)\right)
       \times \\
       &\qquad \left(\mathbf{A} + \frac{1}{z-z_{2}} \left( 
       (z_{2} - \zeta_{1}) \frac{\undertilde{\mathbf{p}}_{1} \mathbf{q}^{\dag}_{2}
       }{\mathbf{q}^{\dag}_{2} \mathbf{A}^{-1} \undertilde{\mathbf{p}}_{1}} + 
       (\zeta_{1} - \zeta_{2}) \frac{\mathbf{A}^{-1} \mathbf{p}_{1} 
       \mathbf{q}^{\dag}_{2}
       }{\mathbf{q}^{\dag}_{2} \mathbf{A}^{-2} \mathbf{p}_{1}}\right)\right).
    \end{align*}
    \item The equations of motion $(\mathbf{Q},\tilde{\mathbf{Q}}) = 
    \Phi(\undertilde{\mathbf{Q}},\mathbf{Q})$ have the implicit form
    \begin{align}
      \left( 
      (z_{2} - \zeta_{1}) \frac{\mathbf{A}^{-1}\undertilde{\mathbf{p}}_{1} 
      }{\mathbf{q}^{\dag}_{2} \mathbf{A}^{-1} \undertilde{\mathbf{p}}_{1}} + 
      (\zeta_{1} - \zeta_{2}) \frac{\mathbf{A}^{-2} \mathbf{p}_{1} 
      }{\mathbf{q}^{\dag}_{2} \mathbf{A}^{-2} \mathbf{p}_{1}}\right) &= 
      \left( (z_{2}- z_{1}) 
      \frac{\mathbf{p}_{1} }{\mathbf{q}^{\dag}_{2} \mathbf{p}_{1}} + 
      (z_{1} - \zeta_{2}) \frac{\mathbf{A}^{-1} \tilde{\mathbf{p}}_{1} 
      }{\mathbf{q}^{\dag}_{2} \mathbf{A}^{-1} \tilde{\mathbf{p}}_{1}}\right)\label{eq:eqmot-p}\\
      \left(
       (z_{1} - \zeta_{2}) \frac{\undertilde{\mathbf{q}}^{\dag}_{2} \mathbf{A}^{-1}
       }{\undertilde{\mathbf{q}}^{\dag}_{2} \mathbf{A}^{-1} \mathbf{p}_{1}} + 
       (\zeta_{2} - \zeta_{1}) \frac{ \mathbf{q}^{\dag}_{2}
        \mathbf{A}^{-2}
       }{\mathbf{q}^{\dag}_{2} \mathbf{A}^{-2} \mathbf{p}_{1}}\right) &=
       \left( (z_{1} - z_{2}) 
        \frac{\mathbf{q}^{\dag}_{2}}{\mathbf{q}^{\dag}_{2} \mathbf{p}_{1}} + 
        (z_{2} - \zeta_{1}) \frac{\tilde{\mathbf{q}}^{\dag}_{2} \mathbf{A}^{-1}
        }{\tilde{\mathbf{q}}^{\dag}_{2}\mathbf{A}^{-1}\mathbf{p}_{1}}\right),\label{eq:eqmot-q}
    \end{align}
    and since we are only interested in the spaces spanned by $\tilde{\mathbf{p}}_{1}$ and
    $\tilde{\mathbf{q}}^{\dag}_{2}$, we can take $\tilde{\mathbf{p}}_{1}$ and
    $\tilde{\mathbf{q}}^{\dag}_{2}$ to be given by the explicit formulas
    \begin{align*}
      \tilde{\mathbf{p}}_{1} &= \mathbf{A}\left( 
      (z_{1} - z_{2})\frac{\mathbf{p}_{1}}{\mathbf{q}^{\dag}_{2} \mathbf{p}_{1}} + 
      (z_{2} - \zeta_{1}) \frac{\mathbf{A}^{-1}\undertilde{\mathbf{p}}_{1} 
      }{\mathbf{q}^{\dag}_{2} \mathbf{A}^{-1} \undertilde{\mathbf{p}}_{1}} + 
      (\zeta_{1} - \zeta_{2}) \frac{\mathbf{A}^{-2} \mathbf{p}_{1} 
      }{\mathbf{q}^{\dag}_{2} \mathbf{A}^{-2} \mathbf{p}_{1}}\right), \\
      \tilde{\mathbf{q}}^{\dag}_{2} &= \left(
      (z_{2} - z_{1}) \frac{\mathbf{q}^{\dag}_{2}}{\mathbf{q}^{\dag}_{2} \mathbf{p}_{1}} + 
      (z_{1} - \zeta_{2}) \frac{\undertilde{\mathbf{q}}^{\dag}_{2} \mathbf{A}^{-1}
       }{\undertilde{\mathbf{q}}^{\dag}_{2} \mathbf{A}^{-1} \mathbf{p}_{1}} + 
       (\zeta_{2} - \zeta_{1}) \frac{ \mathbf{q}^{\dag}_{2}
        \mathbf{A}^{-2}
       }{\mathbf{q}^{\dag}_{2} \mathbf{A}^{-2} \mathbf{p}_{1}}
      \right) \mathbf{A} .
    \end{align*}
    \item Equations (\ref{eq:eqmot-p}--\ref{eq:eqmot-q}) 
    are the discrete Euler-Lagrange equations~(\ref{eq:discrEL}) with the
    Lagrangian function 
    \begin{align*}
      \mathcal{L}(\mathbf{X},\mathbf{Y}) &= 
      (z_{2} - z_{1}) \log(\mathbf{x}^{\dag}_{2} \mathbf{x}_{1}) + 
      (z_{1} - \zeta_{2}) \log(\mathbf{x}^{\dag}_{2} \mathbf{A}^{-1} \mathbf{y}_{1}) + \\ &\qquad
      (\zeta_{2} - \zeta_{1})\log(\mathbf{y}^{\dag}_{2} \mathbf{A}^{-2} \mathbf{y}_{1}) + 
      (\zeta_{1} - z_{2}) \log(\mathbf{y}^{\dag}_{2}\mathbf{A}^{-1} \mathbf{x}_{1}),
    \end{align*}
    where $\mathbf{X}=(\mathbf{x}_{1}, \mathbf{x}^{\dag}_{2})$,
    $\mathbf{Y} = (\mathbf{y}_{1}, \mathbf{y}^{\dag}_{2}) \in \mathcal{Q}=\mathbb{P}^{r-1} \times 
    (\mathbb{P}^{r-1})^{\dag}$. Every vector is an actual $r$-vector that we consider up 
    to rescaling. This makes $\mathcal{L}$ defined up to an additive constant, but this does not
    affect the Euler-Lagrange equations.
  \end{enumerate}
\end{theorem}
\begin{proof}
  Parts (i)  follows immediately from equations~(\ref{eq:q1_ptp1qtq2}--\ref{eq:tp2_tpt1qtq2})
  in Theorem~\ref{thm:flip}. To establish parts (ii) and (iii), consider 
  the conjugated 
  momentum, $\mathbf{P} = \langle \mathbf{\pi}^{\dag}_{1}, \mathbf{\pi}_{2} \rangle = 
  \dfrac{\partial \mathcal{L}}{\partial 
  \mathbf{Y}}(\undertilde{\mathbf{Q}},\mathbf{Q})$. The discrete Euler-Lagrange 
  equations~(\ref{eq:discrEL}) then split into two groups,
  \begin{equation*}
    \mathbf{\pi}^{\dag}_{1} =  \dfrac{\partial \mathcal{L}}{\partial 
    \mathbf{\mathbf{y}_{1}}}(\undertilde{\mathbf{Q}},\mathbf{Q}) = - 
    \dfrac{\partial\mathcal{L}}{\partial
     \mathbf{x}_{1}}(\mathbf{Q},\widetilde{\mathbf{Q}})\qquad\text{and}\qquad
     \mathbf{\pi}_{2} = \dfrac{\partial \mathcal{L}}{\partial 
     \mathbf{\mathbf{y}^{\dag}_{2}}}(\undertilde{\mathbf{Q}},\mathbf{Q}) = - 
     \dfrac{\partial\mathcal{L}}{\partial
      \mathbf{x}^{\dag}_{2}}(\mathbf{Q},\widetilde{\mathbf{Q}}).
  \end{equation*}
   The first equation becomes
  \begin{align*}
    \mathbf{\pi}^{\dag}_{1} &= 
    (z_{1} - \zeta_{2}) \frac{\undertilde{\mathbf{q}}^{\dag}_{2} \mathbf{A}^{-1}
     }{\undertilde{\mathbf{q}}^{\dag}_{2} \mathbf{A}^{-1} \mathbf{p}_{1}} + 
     (\zeta_{2} - \zeta_{1}) \frac{ \mathbf{q}^{\dag}_{2}
      \mathbf{A}^{-2}
     }{\mathbf{q}^{\dag}_{2} \mathbf{A}^{-2} \mathbf{p}_{1}}
     =     (z_{1} - z_{2}) 
       \frac{\mathbf{q}^{\dag}_{2}}{\mathbf{q}^{\dag}_{2} \mathbf{p}_{1}} + 
       (z_{2} - \zeta_{1}) \frac{\tilde{\mathbf{q}}^{\dag}_{2} \mathbf{A}^{-1}
       }{\tilde{\mathbf{q}}^{\dag}_{2}\mathbf{A}^{-1}\mathbf{p}_{1}},
  \end{align*}
  which, on one hand, is equation~(\ref{eq:eqmot-q}), and on the other hand, is the equality of
  two different expression for $\mathbf{q}^{\dag}_{1}$ in $\mathbf{B}^{\mathbf{A}}_{1}(z)$, one
  coming from  $\undertilde{\mathbf{B}}_{2}^{\mathbf{A}}(z) 
  \undertilde{\mathbf{B}}_{1}^{\mathbf{A}}(z) = \mathbf{B}_{1}^{\mathbf{A}}(z) 
  \mathbf{B}_{2}^{\mathbf{A}}(z)$, and the other from $  \mathbf{B}_{2}^{\mathbf{A}}(z)
  \mathbf{B}_{1}^{\mathbf{A}}(z) = \tilde{\mathbf{B}}_{1}^{\mathbf{A}}(z) 
  \tilde{\mathbf{B}}_{2}^{\mathbf{A}}(z)$. The other equation is similar, and that 
  completes the proof.
\end{proof}

%

\section{The Isomonodromic Case}\label{sec:isom_dyn} 
In this section we consider in detail an example of an elementary
isomonodromy transformation defined in Section 3 of \cite{Kri:2004:ATDEWRECRP}.
We show that this transformation can be written in the Lagrangian form and then verify that,
similarly to the polynomial case considered in \cite{AriBor:2006:MSDDPE},
for rank $r=2$  matrices $\mathbf{L}(z)$ whose
divisor has $2$ simple zeroes and $2$ simple poles, this transformation, 
when written in the \emph{spectral coordinates}
$p$ and $q$, reduces to the
difference Painlev\'e equation dPV of the Sakai's hierarchy \cite{Sak:2001:RSAWARSGPE}. 
Thus, we establish that dPV can be written in the Lagrangian form.

\subsection{The spectral coordinates}\label{ssec:the_spectral_coordinates} 
Let $\mathbf{L}(z)\in \mathcal{M}^{\mathcal{D}}_{r}$, where 
$\mathcal{D}=\sum_{i}(z_{i} - \zeta_{i})$, and all points are finite, 
distinct and do not differ by an integer. Since the isomonodromy equations~(\ref{eq:dEQ})
are invariant w.r.t.~the conjugation action of the gauge group $\operatorname{GL}_{r}(\mathbb{C})$,
we can use this action to diagonalize $\mathbf{L}_{0}$, 
$\mathbf{L}_{0} = \operatorname{diag}\{\rho_{1},\dots,\rho_{r}\}$, which
reduces the gauge group to the subgroup
$\operatorname{D}_{r}\subset \operatorname{GL}_{r}$ of the diagonal matrices. Next, consider some
 asymptotic properties of $\mathbf{L}(z)$. Namely, let us
first introduce the matrix
$\mathbf{L}_{\infty}:=-\operatorname{res}_{\infty} \mathbf{L}(z)\,dz
= \sum_{k}\mathbf{L}_{k}$ and put $k_{i} :=
\frac{1}{\rho_{i}}(\mathbf{L}_\infty)_{ii} = (\mathbf{L}_{0}^{-1}
\mathbf{L}_{\infty})_{ii}$. Following \cite{AriBor:2006:MSDDPE},
we define the \emph{type} of $\mathbf{L}(z)$ as follows.

\begin{definition} The type $\theta$ of the matrix $\mathbf{L}(z)$
is the following collection of parameters:
\begin{equation*}
  \label{eq:typeTheta}
  \theta(\mathbf{L}(z))=\left\{z_{1},\dots,z_{n};\zeta_{1},\dots,\zeta_{n};
  \rho_{1},\dots,\rho_{r};k_{1},\dots,k_{r}\right\}. 
\end{equation*}
From the multiplicative representation (\ref{eq:L-mult}) we see that
these parameters are not independent, since
\begin{equation}
  \label{eq:k-rel}
  k_{1} +\cdots + k_{r} =
  \operatorname{tr}\mathbf{L}^{-1}_{0}\mathbf{L}_{\infty} =
  \sum_{i=1}^{n}   \operatorname{tr}(\mathbf{G}_{i} \mathbf{A}^{-1})=
  \sum_{i=1}^{n}    (z_{i} - \zeta_{i}).
\end{equation}
For a general choice of parameters, this is the only relation. We
denote by ${\mathcal M}_{r}^{\theta}$ the space of matrices of type
$\theta$ and rank $r$. This space is clearly invariant under the conjugation by
non-degenerate diagonal matrices. Factoring out this action 
we obtain the coarse moduli space $\widehat{{\mathcal M}}^{\theta}_{r}$.
\end{definition}

It is worth mentioning that fixing the type $\theta$ of $\mathbf{L}(z)$ corresponds to 
considering a symplectic leaf of the canonical foliation of $\mathcal{M}^{\mathcal{D}}_{r}$ 
w.r.t.~the universal algebro-geometric symplectic form $\omega$ of 
Krichever and Phong \cite{KriPho:1997:IGSESGT,KriPho:1998:SFTS},
see the survey \cite{DHoKriPho::STSFHTS} for details.

\begin{lemma} The dimension of the big cell of the moduli space
  $\widehat{{\mathcal M}}^{\theta}_{r}$ is
\begin{equation*}
  \label{eq:dim-Mtheta}
  \dim {\mathcal M}^{\theta}_{r} = 2(n-1)(r-1).
\end{equation*}
\end{lemma}
\begin{proof}
  Consider the multiplicative representation of $\mathbf{L}(z)$. Each
  elementary divisor is given by $2(r-1)$ parameters. Fixing the diagonal
  elements of $\mathbf{L}_{\infty}$ imposes $r-1$ (in view of
  (\ref{eq:k-rel})) additional conditions, and therefore 
  \begin{equation*}
    \dim{\mathcal M}^{\theta}_{r} = n(2r-2) - (r-1) = (2n-1)(r-1).
  \end{equation*}
  Further action by the diagonal matrices reduces the dimension by 
  $r-1$, and so 
  \begin{equation*}
    \dim\widehat{{\mathcal M}}^{\theta}_{r} = 2(n-1)(r-1).
  \end{equation*}
 \end{proof}

Thus, when $n=r=2$, $\widehat{{\mathcal M}}^{\theta}_{r}$ is a complex surface. 
The spectral coordinate system $(q,p)$ on the space $\widehat{{\mathcal
M}}^{\theta}_{2}$ is given by the \emph{zero} $q$ of 
$\mathbf{L}(z)_{12}$ and the 
value of $\mathbf{L}(z)_{11}$ at $q$ normalized in the following way:
\begin{equation*}
  \label{eq:pd-def}
  p =
  \frac{q-z_{1}}{q-\zeta_{2}}\mathbf{L}(q)_{11},\qquad\text{where}\quad
  \mathbf{L}(q)_{12}=0. 
\end{equation*}

The spectral coordinates $p$ and $q$ are essentially the same as in \cite{AriBor:2006:MSDDPE}.
From the integrable systems point of view these coordinates are a particular case of the 
Darboux coordinates of the universal algebro-geometric symplectic form $\omega$. For a particular
case of the hyperelliptic KdV curves these coordinates were first considered by Novikov and 
Veselov \cite{VesNov:1982:PBTCWAGDKVEFP}, the general case was recently established by 
Krichever in \cite{Kri:2000:EATL, Kri:2000:ESDNENBAE}.

Note that since 
\begin{equation*}
  \det \mathbf{L}(z) = \det \mathbf{L}_{0}\frac{(z-\zeta_{1})(z-\zeta_{2})}{(z - z_{1})(z-z_{2})},
  \qquad \mathbf{L}(q) = \begin{bmatrix}
    \frac{p(q - \zeta_{2})}{(q - z_{1})} & 0 \\ * & 
    \frac{\rho_{1} \rho_{2}(q-\zeta_{1})}{p(q - z_{2})}
  \end{bmatrix}.
\end{equation*}

Next, we need to obtain the explicit formulas for the additive and the 
multiplicative representations of $\mathbf{L}(z)$ in the spectral coordinates.


\subsection{The Additive and the Multiplicative Representations of $\mathbf{L}(z)$.}
\label{ssec:the_additive_and_the_multiplicative_representations_of_mathbf_} 

\begin{lemma} The additive representation~(\ref{eq:L-addrep}) of $\mathbf{L}(z)$ in the
$pq$-coordinates is given by 
\begin{align*}
  \mathbf{L}_{1} &=  \frac{q - z_{1}}{z_{2} - z_{1}} 
    \begin{bmatrix}
      1 \\ \rho_{2}(q-z_{2}+k_{2}) - \frac{\rho_{1}\rho_{2}}{p}(q-\zeta_{1})
    \end{bmatrix}\begin{bmatrix}
      \rho_{1}(q-z_{2}+k_{1}) - \frac{p(q-z_{2})(q-\zeta_{2})}{q-z_{1}}& 1
    \end{bmatrix},\\
  \mathbf{L}_{2} &= \frac{q-z_{2}}{z_{1}-z_{2}} 
    \begin{bmatrix}
      1 \\ \rho_{2}(q-z_{1}+k_{2}) -
      \frac{\rho_{1}\rho_{2}(q-z_{1})(q-\zeta_{1})}{p(q-z_{2})}
    \end{bmatrix}\begin{bmatrix}
      \rho_{1}(q-z_{1}+k_{1}) - p(q-\zeta_{2}) & 1
    \end{bmatrix}.
\end{align*}
  
\end{lemma}

\begin{proof}
By choosing a linear normalization in which the first components of the column vectors and the 
second components of the row vectors are equal to $1$, we get
\begin{equation*}
  \mathbf{L}(z) =
  \begin{bmatrix}
    \rho_{1} & 0 \\ 0 & \rho_{2}
  \end{bmatrix} + \alpha_{1} \frac{
    \begin{bmatrix}
      1 \\a^{1}
    \end{bmatrix}
    \begin{bmatrix}
      b_{1} & 1
    \end{bmatrix}
  }{z-z_{1}}   + \alpha_{2} \frac{
    \begin{bmatrix}
      1 \\a^{2}
    \end{bmatrix}
    \begin{bmatrix}
      b_{2} & 1
    \end{bmatrix}
  }{z-z_{2}}.   
\end{equation*}
We then obtain the following equations for
$\alpha_{1}$, $\alpha_{2}$, $b_{1}$, $b_{2}$:
\begin{align*}
  \mathbf{L}(q)_{12} &= \frac{\alpha_{1}}{q-z_{1}} +
  \frac{\alpha_{2}}{q-z_{2}}=0, \\
  p &= \frac{q-z_{1}}{q-\zeta_{2}}\mathbf{L}(q)_{11} =
  \frac{q-z_{1}}{q-\zeta_{2}} \left( \rho_{1} +
    \frac{\alpha_{1}}{q-z_{1}}b_{1} +
    \frac{\alpha_{2}}{q-z_{2}}b_{2}\right)\\ &=
  \frac{q-z_{1}}{q-\zeta_{2}} \left( \rho_{1} +
    \frac{\alpha_{1}}{q-z_{1}}(b_{1}-b_{2}) \right), \\
  (\mathbf{L}_{\infty})_{11}&=\alpha_{1}b_{1} + \alpha_{2}b_{2} =
  \rho_{1} k_{1}.
\end{align*}
Note that we have only three equation for the four unknowns. This is due to  
the conjugation action by the constant diagonal matrices
$\mathbf{D}=\operatorname{diag}\{\gamma_{1},\gamma_{2}\}$. Taking this action into account, 
\begin{align*}
  \mathbf{D}(\alpha_{i}\mathbf{a}_{i}\mathbf{b}_{i}^{\dag})\mathbf{D}^{-1}
  = \left(\frac{\gamma_{1}}{\gamma_{2}}\alpha_{i}\right)
  \begin{bmatrix}
    1 \\ \left( \frac{\gamma_{2}}{\gamma_{1}} a^{i}\right)
  \end{bmatrix}
  \begin{bmatrix}
    \left( \frac{\gamma_{2}}{\gamma_{1}}b_{i}\right) & 1
  \end{bmatrix},
\end{align*}
we see that we can provisionally put $\alpha_{1}=1$. Then
\begin{align*}
  \alpha_{2} &=-\frac{q-z_{2}}{q-z_{1}}\\
  b_{1} &= \rho_{1}k_{1} - \alpha_{2}b_{2} = \rho_{1}k_{1} +
  \frac{q-z_{2}}{q-z_{1}}b_{2} 
  = b_{2} + (q-\zeta_{2})p - (q-z_{1})\rho_{1},
  \intertext{and so }
  b_{2} &= \frac{q-z_{1}}{z_{2}-z_{1}}\left( \rho_{1}(q-z_{1}+k_{1}) -
    p(q-\zeta_{2})\right),\\
  b_{1} &= \frac{q-z_{1}}{z_{2}-z_{1}}\left( \rho_{1} (q-z_{2} +
    k_{1}) - p\frac{(q-\zeta_{2})(q-z_{2})}{q-z_{1}}\right).
\end{align*}

To find $a^{1}$ and $a^{2}$ we use the definition of $k_{2}$ and
  $\det\mathbf{L}(q)$:
  \begin{align*}
    (\mathbf{L}_{\infty})_{22} &=\alpha_{1} a^{1} + \alpha_{2} a^{2} =
    \rho_{2} k_{2} \\
      \det \mathbf{L}(q) &= 
      \det\begin{bmatrix}
        \frac{q-\zeta_{2}}{q-z_{1}}p & 0 \\ * & \rho_{2} + \frac{\alpha_{1}}{q-z_{1}}a^{1} +
        \frac{\alpha_{2}}{q-z_{2}}a^{2}
      \end{bmatrix} = p\frac{q-\zeta_{2}}{q-z_{1}}\left(\rho_{2} +
        \frac{a^{1}-a^{2}}{q-z_{1}} \right) \\& = \rho_{1}\rho_{2}\frac{
        (q-\zeta_{1})(q-\zeta_{2})}{ (q-z_{1}) (q-z_{2})},
      \intertext{and so}
      a^{1}-a^{2} &=
      \frac{\rho_{1}\rho_{2}}{p}\cdot\frac{(q-\zeta_{1})(q-z_{1})}{q-z_{2}} -\rho_{2}(q-z_{1}).
  \end{align*}
  Thus,
  \begin{align*}
    a^{1}&=\rho_{2}k_{2} -\alpha_{2} a^{2} = \rho_{2}k_{2} +
    \frac{q-z_{2}}{q-z_{1}}a^{2} = a^{2} +
    \frac{\rho_{1}\rho_{2}}{p}\cdot\frac{(q-\zeta_{1})(q-z_{1})}{q-z_{2}}
    -\rho_{2}(q-z_{1}),\\
    \intertext{which results in}
  a^{1} &= \frac{q-z_{1}}{z_{2}-z_{1}}\left( \rho_{2}(q-z_{2}+k_{2})
    -\frac{\rho_{1}\rho_{2}(q-\zeta_{1})}{p}\right), \\
    a^{2} &= \frac{q-z_{1}}{z_{2}-z_{1}}\left( \rho_{2}(q-z_{1}+k_{2}) -
    \frac{\rho_{1}\rho_{2}}{p}\cdot\frac{(q-\zeta_{1})(q-z_{1})}{q-z_{2}}\right).
  \end{align*}
  Conjugating by the diagonal matrix
  $\mathbf{D}=\operatorname{diag}\{q-z_{1},z_{2}-z_{1}\}$ finishes the
  proof. 
\end{proof}

To obtain the multiplicative description of $\mathbf{L}(z)$ in the spectral coordinates, 
we need the following Lemma.

\begin{lemma}\label{lem:add2mult}
  Let 
  \begin{equation*}
    \mathbf{L}(z) = \mathbf{L}_{0} + \frac{\mathbf{L}_{1}}{z-z_{1}} + 
    \frac{\mathbf{L}_{2}}{z-z_{2}} = \mathbf{B}^{\mathbf{A}}_{1}(z) \mathbf{B}^{\mathbf{A}}_{2}(z),
  \end{equation*}
  where $\mathbf{L}_{0} = \mathbf{A}^{2}$. Then
  \begin{enumerate}[(i)]
    \item the additive representation is given in terms of the multiplicative representation 
    by 
    \begin{equation*}
      \mathbf{L}_{1} = \mathbf{G}_{1} \mathbf{B}^{\mathbf{A}}_{2}(z_{1}),\qquad
      \mathbf{L}_{2} = \mathbf{B}^{\mathbf{A}}_{1}(z_{2})\mathbf{G}_{2};
    \end{equation*}
    \item the multiplicative representation is given in terms of the additive representation 
    by
    \begin{equation*}
      \mathbf{G}_{1}\mathbf{A} = \mathbf{L}_{1} + \frac{(z_{1} - \zeta_{1}) - 
      \operatorname{tr}(\mathbf{L}_{0}^{-1}\mathbf{L}_{1})}{
      \operatorname{tr}(\mathbf{L}_{0}^{-1}\mathbf{L}_{1}\mathbf{L}_{2})}
      \mathbf{L}_{1}\mathbf{L}_{2},\qquad
      \mathbf{A}\mathbf{G}_{2} = \mathbf{L}_{2} + \frac{(z_{2} - \zeta_{2}) - 
      \operatorname{tr}(\mathbf{L}_{0}^{-1}\mathbf{L}_{2})}{
      \operatorname{tr}(\mathbf{L}_{0}^{-1}\mathbf{L}_{1}\mathbf{L}_{2})}
      \mathbf{L}_{1}\mathbf{L}_{2}.
    \end{equation*}
  \end{enumerate}
\end{lemma}
\begin{proof}
  Part (i) follows immediately from taking the residues at $z_{i}$. To establish part (ii), 
  note that from part (i) it follows that 
  \begin{equation}
    \mathbf{G}_{1} \mathbf{A} = \mathbf{L}_{1} + 
    \frac{\mathbf{G}_{1} \mathbf{G}_{2}}{z_{2} - z_{1}},\qquad
    \mathbf{A} \mathbf{G}_{2} = \mathbf{L}_{2} + 
    \frac{\mathbf{G}_{1} \mathbf{G}_{2}}{z_{1} - z_{2}},
    \label{eq:a2m-1}
  \end{equation}
  and so what we really need is to establish the identity
  \begin{equation*}
    \mathbf{G}_{1} \mathbf{G}_{2} = (z_{2} - z_{1})
    \frac{(z_{1} - \zeta_{1}) - \operatorname{tr}(\mathbf{L}_{0}^{-1} \mathbf{L}_{1})}{
    \operatorname{tr}(\mathbf{L}_{0}^{-1}\mathbf{L}_{1}\mathbf{L}_{2})} 
    \mathbf{L}_{1} \mathbf{L}_{2} = (z_{1} - z_{2}) 
    \frac{(z_{2} - \zeta_{2}) - \operatorname{tr}(\mathbf{L}_{0}^{-1} \mathbf{L}_{2})}{
    \operatorname{tr}(\mathbf{L}_{0}^{-1}\mathbf{L}_{1}\mathbf{L}_{2})}
    \mathbf{L}_{1} \mathbf{L}_{2}.
  \end{equation*}
  Since $\mathbf{L}_{i}$ and $\mathbf{G}_{i}$ are of rank one, from (i) we see that 
  $\mathbf{G}_{1} \mathbf{G}_{2} = \gamma \mathbf{L}_{1} \mathbf{L}_{2}$. The proportionality
  constant $\gamma$ can be found using either of the equations~(\ref{eq:a2m-1}). For example,
  using the first equation, we see that
  \begin{equation*}
    \gamma = \frac{\operatorname{tr}( \mathbf{G}_{1} \mathbf{G}_{2} 
    \mathbf{A}^{-2})}{\operatorname{tr}(\mathbf{L}_{0}^{-1} \mathbf{L}_{1} \mathbf{L}_{2})} = 
    \frac{(z_{2} - z_{1})\operatorname{tr}( \mathbf{G}_{1} \mathbf{A}^{-1} - 
    \mathbf{L}_{1} \mathbf{A}^{-2})}{
    \operatorname{tr}(\mathbf{L}_{0}^{-1} \mathbf{L}_{1} \mathbf{L}_{2})} = 
    (z_{2} - z_{1})\frac{(z_{1} - \zeta_{1}) - \operatorname{tr}( \mathbf{L}_{0}^{-1}\mathbf{L}_{1})
    }{\operatorname{tr}(\mathbf{L}_{0}^{-1} \mathbf{L}_{1} \mathbf{L}_{2})}.
  \end{equation*} 
\end{proof}
\begin{corollary}\label{cor:L-mult}
  The multiplicative representation of $\mathbf{L}(z)$ in the spectral coordinates is given by 
  \begin{align*}
    \mathbf{L}(z) &= \left( \mathbf{L}_{0} + \frac{\mathbf{G}_{1} \mathbf{A}}{z - z_{1}} \right)
    \mathbf{L}_{0}^{-1} \left( \mathbf{L}_{0 } + \frac{\mathbf{A} \mathbf{G}_{2}}{z - z_{2}}\right),
    \\
    \intertext{where}
    \mathbf{G}_{1} \mathbf{A} &=\frac{p}{p-\rho_{1}} \begin{bmatrix}
      1 \\ \rho_{2} (q - z_{2} + k_{2}) - \frac{\rho_{1}\rho_{2}(q - \zeta_{1})}{p}
    \end{bmatrix} \begin{bmatrix}
    - \rho_{1} (q - k_{1} - \zeta_{2}) + \frac{\rho_{1}^{2}(q - z_{1})}{p}   & 1
    \end{bmatrix} \\
    \mathbf{A} \mathbf{G}_{2} &= \frac{\rho_{1}}{\rho_{1} - p} \begin{bmatrix} 1 \\
      - \rho_{2} (q - k_{2} - \zeta_{1}) + \frac{\rho_{2} p}{\rho_{1}}(q - z_{2})
    \end{bmatrix} \begin{bmatrix}
      \rho_{1} (q + k_{1} - z_{1}) - p(q - \zeta_{2}) & 1
    \end{bmatrix}.
  \end{align*}
\end{corollary}
\begin{proof}
  Proof is a direct calculation.
\end{proof}

\subsection{The isomonodromic Lagrangian and dPV}\label{ssec:the_isomonodromic_lagrangian_and_dpv} 
We are now ready to prove the main Theorem of this section.
\begin{theorem}\label{thm:Isom-Lagr}
  Consider the special isomonodromic transformation of the form
  \begin{equation}
    \mathbf{L}(z)=\mathbf{B}^{\mathbf{A}}_{1}(z) \mathbf{B}^{\mathbf{A}}_{2}(z) \mapsto
    \tilde{\mathbf{L}}(z) = \tilde{\mathbf{B}}^{\mathbf{A}}_{1}(z) 
    \tilde{\mathbf{B}}^{\mathbf{A}}_{2}(z) 
    =\mathbf{B}^{\mathbf{A}}_{2}(z+1) \mathbf{B}^{\mathbf{A}}_{1}(z) = 
    \mathbf{B}^{\mathbf{A}}_{2}(z+1) \mathbf{L}(z) (\mathbf{B}^{\mathbf{A}}_{2}(z))^{-1}.
      \label{eq:dIsomSp-a}
  \end{equation}
  \begin{enumerate}[(i)]
    \item This transformation satisfies the time-dependent Euler-Lagrange equations
    \begin{equation*}
      \frac{\partial \mathcal{L}}{\partial \mathbf{Y}} (\mathbf{Q}_{k-1},\mathbf{Q}_{k},k-1) + 
      \frac{\partial \mathcal{L}}{\partial \mathbf{X}} (\mathbf{Q}_{k},\mathbf{Q}_{k+1},k) = 0
      \label{eq:discrEL-time-a}
    \end{equation*}
    with the Lagrangian function given by 
    \begin{align*}
      \mathcal{L}(\mathbf{X},\mathbf{Y},t) &= 
      (z_{2}(t) - z_{1}) \log(\mathbf{x}^{\dag}_{2} \mathbf{x}_{1}) + 
      (z_{1} - \zeta_{2}(t)) \log(\mathbf{x}^{\dag}_{2} \mathbf{A}^{-1} \mathbf{y}_{1}) + \\ &\qquad
      (\zeta_{2}(t) - \zeta_{1})\log(\mathbf{y}^{\dag}_{2} \mathbf{A}^{-2} \mathbf{y}_{1}) + 
      (\zeta_{1} - z_{2}(t)) \log(\mathbf{y}^{\dag}_{2}\mathbf{A}^{-1} \mathbf{x}_{1}),
    \end{align*}
    where $z_{2}(t) = z_{2} - t$ and $\zeta_{2}(t) = \zeta_{2} - t$.
    \item This transformation defines a birational map $\psi:
    \widehat{\mathcal{M}}^{\theta}_{2}\to\widehat{\mathcal{M}}^{\tilde{\theta}}_{2}$, where 
      $\tilde{z}_{2}=z_{2}-1$, $\tilde{\zeta}_{2}=\zeta_{2}-1$  and all
      other parameters are unchanged. In the spectral coordinates $p$, $q$ on
      $\widehat{\mathcal{M}}^{\theta}_{2}$ and $\tilde{p}$, $\tilde{q}$ on
      $\widehat{\mathcal{M}}^{\tilde{\theta}}_{2}$ this transformation is given by the 
      difference Painlev\'e equation dPV of the Sakai's hierarchy,
      \begin{equation*}
        \label{eq:dPVpq}
        \left\{
            \begin{aligned}
              q + \tilde{q}  & = z_{2} + \tilde{\zeta_{2}} -
              \frac{\rho_{1}(k_{2} + \zeta_{1} - z_{2})}{p-\rho_{1}} +
              \frac{\rho_{2}(k_{2}+\tilde{\zeta}_{2} -
                z_{1})}{p-\rho_{2}}\\
              p\tilde{p} &= \rho_{1} \rho_{2} \frac{(\tilde{q} -
                \zeta_{1})(\tilde{q} -
                z_{1})}{(\tilde{q}-\tilde{\zeta}_{2})(\tilde{q}-\tilde{z}_{2}) }
            \end{aligned}
          \right.\,.
      \end{equation*}   
  \end{enumerate}
\end{theorem}
\begin{proof}
  Part (i) is proved in exactly the same way as Theorem~\ref{thm:EQ-M} part (iii). To establish 
  part (ii), it suffices to rewrite equation~(\ref{eq:dIsomSp-a}) in the form 
  $(\mathbf{B}^{\mathbf{A}}_{2}(z+1))^{-1} \tilde{\mathbf{L}}(z) = \mathbf{B}^{\mathbf{A}}_{1} (z)$
  and then evaluate it at $\tilde{q}$. From the definition of $\tilde{q}$ and $\tilde{p}$, we
  get 
  \begin{equation*}
    \mathbf{L}_{0}^{-1}\left(\mathbf{L}_{0} - \frac{\mathbf{A} \mathbf{G}_{2}}{\tilde{q} - 
    \tilde{\zeta}_{2}}\right) \mathbf{A}^{-1} \begin{bmatrix}
      \frac{\tilde{p}(\tilde{q} - \tilde{\zeta}_{2})}{(\tilde{q} - \tilde{z}_{1})} & 0 \\
      * & \frac{\rho_{1} \rho_{2} (\tilde{q} - \tilde{\zeta}_{1})}{\tilde{p}(\tilde{q} - 
      \tilde{z}_{2})}
    \end{bmatrix} = \left(\mathbf{L}_{0} + \frac{\mathbf{G}_{1} \mathbf{A}}{\tilde{q} 
    - z_{1}}\right) \mathbf{A}^{-1}.
  \end{equation*}
  Thus, second columns of the matrices $\displaystyle \mathbf{L}_{0}^{-1}\left(\mathbf{L}_{0} - 
  \frac{\mathbf{A} \mathbf{G}_{2}}{\tilde{q} - 
  \tilde{\zeta}_{2}}\right)$ and $\displaystyle \left(\mathbf{L}_{0} + 
  \frac{\mathbf{G}_{1} \mathbf{A}}{\tilde{q} - z_{1}}\right)$ are proportional with the
  proportionality coefficient $\displaystyle \frac{\rho_{1} \rho_{2} (\tilde{q} - 
  \tilde{\zeta}_{1})}{\tilde{p}(\tilde{q} - \tilde{z}_{2})}$. Direct calculation then completes 
  the proof.
\end{proof}

\thanks{I am very grateful to I.~Krichever for many helpful and stimulating discussions. 
This research was supported in part by the University of Northern Colorado Summer~2006 SPARC Small Grant
Assistance Program.} 
 
\small
\bibliographystyle{amsalpha}
\providecommand{\bysame}{\leavevmode\hbox to3em{\hrulefill}\thinspace}
\providecommand{\MR}{\relax\ifhmode\unskip\space\fi MR }
\providecommand{\MRhref}[2]{%
  \href{http://www.ams.org/mathscinet-getitem?mr=#1}{#2}
}
\providecommand{\href}[2]{#2}

\end{document}